\documentclass[a4paper]{article} 
\usepackage{amsmath,amsthm,amssymb}

\title{Quantization of classical integrable systems \\
Part IV: systems of resonant oscillators}
\author{M. Marino and N. N. Nekhoroshev \\
{\small Dipartimento di Matematica, Universit\`{a} degli Studi di Milano,} \\
{\small via Saldini 50, I-20133 Milano (Italy)}}

\newtheorem{thm}{Theorem}

\newtheorem{prop}[thm]{Proposition}

\theoremstyle{definition}
\newtheorem{defn}{Definition}[section]

\theoremstyle{remark}

\def\beq{\begin{equation}}
\def\eeq{\end{equation}}
\def\im{{\rm Im}\,}
\def\re{{\rm Re}\,}

\def\aop{{\cal A}}
\def\bop{{\cal B}}

\def\rn{\mathbb{R}}
\def\nn{\mathbb{N}}
\def\zn{\mathbb{Z}}
\def\cn{\mathbb{C}}

\def\ipn2{\left[\frac n2\right]}

\numberwithin{equation}{section} \numberwithin{thm}{section}

\begin{document}

\maketitle

\begin{abstract}
By applying methods already discussed in a previous series of
papers by the same authors, we construct here classes of
integrable quantum systems which correspond to $n$ fully resonant
oscillators with nonlinear couplings. The same methods are also
applied to a series of nontrivial integral sets of functions,
which can be constructed when additional symmetries are present
due to the equality of some of the frequencies. Besides, for $n=3$
and resonance 1:1:2, an exceptional integrable system is obtained,
in which integrability is not explicitly connected with this type
of symmetry. In this exceptional case, quantum integrability can
be realized by means of a modification of the symmetrization
procedure.
\end{abstract}

\section{Introduction}

In \cite{part3} we have already given examples of applications to
concrete systems of a general procedure \cite{part2}, with which a
classical integrable system can generally be transformed into a
quasi-integrable quantum system \cite{part1} (see also references
therein). In this paper we consider systems describing an
arbitrary number $n$ of oscillators with a fully resonant set of
frequencies. More exactly, in the classical case we describe a
class of integrable sets of functions $F$ such that $F_1= l_1I_1+
\dots+ l_n I_n$, where $I_i= \frac 12 (p_i^2 +x_i^2)$ for $i=1,
\dots, n$, and $l_1, \dots, l_n$ are nonzero integer values, i.e.,
$l\in (\zn \setminus \{0\})^n$ (one can assume that they have no
common divisors other than 1). Hence the set $F$ contains the
hamiltonian $H=F_1$ of a set of fully resonant linear oscillators.
We are able to construct a whole class of integrable sets of
functions by exploiting the condition of resonance. Furthermore,
when some of the frequencies are equal to each other, it is
possible to construct a wider class of nontrivial integrable
systems, by exploiting the symmetry provided by these equalities.

These integrable sets contain in general $2n-k$ elements, where
$k$ is equal to the number of elements in the central subset
\cite{TMMO}. This number, for the various integral sets here
considered, can take all possible values from 1 to $n$. Each
integrable set can be applied to all systems whose hamiltonian is
an arbitrary function of the central elements, and so to a whole
class of systems which describe nonlinear oscillations with
completely resonant frequencies of small oscillations. It means
that if the amplitude of oscillations tends to zero, then in this
limit the oscillations tend to the linear ones with a completely
resonant set of frequencies. The situation $n=k=3$,
$|h_1|=|h_2|\neq 0$, is studied in particular detail. For the
special case $|h_1|=|h_2|=1$, $|h_3|=2$, we are able to construct
an exceptional integrable system which is not explicitly based on
the symmetry connected with the two equal frequencies.

Using the results of
\cite{part2}, we then construct the quantum analogues of all these
classical integrable systems. For the exceptional system, a
modification of the symmetrization procedure is required in order
to obtain a quasi-integrable quantum system. This modification
consist in the addition of a lower order term to the direct
symmetrization of one of the classical functions, and thus is of
the same type as the modification that was required in
\cite{part3} for the integration of the free quantum rigid body in
6-dimensional space.

The investigations of the integrable systems considered in
\cite{part3} and in the present paper illustrate in concrete
situations the general concepts presented in \cite{part1} and
\cite{part2}. Note however that the systems considered in
\cite{part3}, describing the motion of particles in a central
force field and the free rotation of a rigid body, have
application in physics only when the dimension $n$ of
configuration space is equal to 3. On the contrary, the systems of
oscillators which are considered here are interesting for
applications for any not too large value of $n$. For example,
spectral lattices which are constructed from quantum integrable
systems are useful for studying data of spectrography in physics
and molecular chemistry. The points of such lattices are vectors
$\lambda=(\lambda_1, \ldots, \lambda_k)$, such that each vector is
made of eigenvalues $\lambda_i$, $i=1,\ldots,k$, of the central
operators ${\cal F}_1,\ldots,{\cal F}_k$ of an integrable set
${\cal F}=({\cal F}_1,\ldots,{\cal F}_k;{\cal F}_{k+1}, \ldots,
{\cal F}_{2n-k})$ of linear differential operators, and these
eigenvalues correspond to common eigenfunctions $\psi$ of these
operators: ${\cal F}_i\psi=\lambda_i\psi$, $i=1,\ldots,k$. More
details can be found for example in \cite{NSZhilin,NSZ,drobrez}.
Sets of nonlinear fully resonant oscillators also play a central
role in the study of important infinite-dimensional systems
\cite{bamb, Trud}.

\section{Classical oscillators with completely resonant set of
frequencies} \label{closc}

It is easy to construct a linear basis of the infinite-dimensional
linear space of all real polynomial functions in $(x,p)$ which are
in involution with $F_1= l_1 I_1+ \dots+ l_n I_n$. To this
purpose, it is useful to introduce the functions
\begin{equation}\label{zxp}
z_j =\frac{x_j +ip_j}{\sqrt 2}\,, \qquad \bar z_j= \frac{x_j
-ip_j}{\sqrt 2}\,,
\end{equation}
so that $\bar z_j$ is the complex conjugate of $z_j$. A linear
basis in the space of all polynomials in $(x,p)$ is obviously
given by the real and imaginary parts of all monomials in $(z,
\bar z)$, i.e., by the functions of the form $\re P_{a,b}$ and
$\im P_{a,b}$, where $P_{a,b}:= z_1^{a_1} \bar z_1^{b_1} z_2^{a_2}
\bar z_2^{b_2} \cdots z_n^{a_n} \bar z_n^{b_n}$ and $a=(a_1, a_2,
\ldots, a_n) \in \zn_+^{n}$, $(b_1, b_2, \ldots, b_n) \in
\zn_+^{n}$. We have
\begin{equation}
\{z_i,z_j\}=\{\bar z_i,\bar z_j\}= 0\,, \qquad \{z_i,\bar z_j\}=
i\delta_{ij} \label{fpb}
\end{equation}
for $i,j=1,\ldots,n$, and $I_i = z_i\bar z_i$. It follows that
\begin{equation} \label{f1pab}
\{F_1, P_{a,b}\} = -i l\cdot (a-b)\, P_{a,b}\,,
\end{equation}
where $l\cdot (a-b)=\sum_{j=1}^n l_j (a_j-b_j)$. We have
$(P_{a,b})^*=P_{b,a}$, where $^*$ denotes complex conjugation.
From (\ref{f1pab}) one thus obtains
\begin{equation}\label{f1rpab}\begin{split}
\{F_1, \re P_{a,b}\}&= l\cdot(b-a) \im P_{a,b}\,, \\
\{F_1, \im P_{a,b}\}&= -l\cdot(b-a) \re P_{a,b} \,. \end{split}
\end{equation}
We see therefore that a basis of the linear space of all
polynomials in involution with $F_1$ is given by the set ${\cal
P}_l:= \big(\re P_{a,b}, \im P_{a,b},\ l\cdot(b-a)=0\big)$.

\subsection{General symmetries}\label{gs} It is easy to find
$2n-1$ functionally independent elements of ${\cal P}_l$, for any
set of frequencies $l=(l_1, \dots, l_n)\in (\zn \setminus
\{0\})^n$. Hence, the system with hamiltonian function $F_1$ is
integrable with number of central integrals $k=1$. For example, if
$l_i \geq 1$ for $i=1, \dots, n$, then it is easy to check that
the function $F_1$ is involution with the functions of the set
$F=(F_1, I_2, I_3, \dots, I_n, R_{12}, R_{13}, \dots, R_{1n})$,
where $R_{ij}:= {\rm Im} \big( z_i^{l_j} \bar z_j^{l_i}\big)$ for
$1\leq i<j\leq n$. If $l_1$ and $l_j$ have a common divisor $d$,
then one can divide both numbers by $d$ and consider instead of
$R_{ij}$ the polynomial $R'_{ij}:= {\rm Im}\big(z_i^{l'_j}\bar
z_j^{l'_i}\big)$, where $l'_i=l_i/d$, $l'_j=l_j/d$. Moreover, it
is easy to see that the functions of the set $F$ are functionally
independent, so that this set is integrable with one central
integral $F_1$. An integrable set $F$ of similar form can also be
constructed when the hamiltonian function $F_1$ has frequencies
$l_1, \dots, l_n$ of different signs. In this case, one takes
\begin{equation} \label{rij}
R_{ij}:= \begin{cases} {\rm Im} \big(z_i^{|l_j|} \bar
z_j^{|l_i|}\big)\qquad &\text{if }l_i l_j>0\,, \\
{\rm Im}\big(z_i^{|l_j|}z_j^{|l_i|}\big)\qquad &\text{if } l_i
l_j<0\,.\end{cases}
\end{equation}

It is also possible to find other integrable sets of functions,
with number $k$ of central integrals ranging from 2 up to the
maximum possible value $n$. As we already explained in
\cite{part3} about a similar situation for one-particle systems
with $SO(n)$ symmetry (see there proposition 
2.4 and related comments), the existence of such sets with higher
$k$ allows one to construct a wider class of integrable systems,
by taking as hamiltonian any arbitrary function of the central
elements of an integrable set. In the present case, functions of
this type can often be represented as polynomials in $(x,p)$ of
degree $>2$. In such cases, they can be considered as
perturbations to the hamiltonian $F_1$ of resonant oscillators,
which become negligible in the limit of small amplitude of
oscillations.

We are now going to describe a general procedure for the
construction of a remarkable class of integral sets of functions
of various $k$. For any $m\in \zn^n$, consider the function
\begin{equation} \label{rm}
R_m= \zeta_1^{m_1} \zeta_2^{m_2}\cdots \zeta_n^{m_n} \,,
\end{equation}
where
\begin{equation} \label{zeta}
\zeta_j^{m_j}:= \begin{cases} z_j^{m_j} &\text{if }\ m_j\geq 0\,,  \\
\bar z_j^{-m_j} &\text{if }\ m_j< 0\,. \end{cases}
\end{equation}
Note that $(\zeta_j^{m_j})^* = \zeta_j^{-m_j}$, so that
$R_m^*=R_{-m}$. For $r\in \zn^n$, let us define
\begin{equation} \label{j}
J_r= \sum_{j=1}^n r_j I_j\,,
\end{equation}
so that $F_1= J_l$. Similarly to (\ref{f1rpab}) we have
\begin{equation} \label{jr}
\{J_r, \re R_m\}= r\cdot m\, \im R_m \,,\qquad \{J_r, \im R_m\}=
-r\cdot m \,\re R_m \,.
\end{equation}
Consider the sets of functions ${\cal A}_l= (\re R_m,\ m\cdot
l=0)$, ${\cal B}_l= (\im R_m, \ m\cdot l=0)$, and $I=(I_1, I_2,
\dots, I_n)$. It is easy to see that ${\cal U}_l:= (I, {\cal A}_l,
{\cal B}_l)$ is a set of functions in involution with $F_1$, which
is closed with respect to Poisson brackets. This means that the
Poisson bracket of two elements of ${\cal U}_l$ can always be
expressed as a function of other elements of ${\cal U}_l$. Note
also that
\begin{equation}\label{imre}
(\re R_m)^2+ (\im R_m)^2= |R_m|^2= I_1^{|m_1|} I_2^{|m_2|} \cdots
I_n^{|m_n|}\,.
\end{equation}

Put
\begin{equation}\label{f1n}
F_i=J_{r^{(i)}}\quad \text{for }\ i=1, \dots,n\,,
\end{equation}
where $r^{(1)}, \dots, r^{(n)}$ are $n$ linearly independent
elements of $\zn^n$, with $r^{(1)}=l$. Obviously $F^{(n)}= (F_1,
\dots, F_n)$ is an integrable set with $k=n$. It is however
possible to obtain other integrable sets with $k$ central
elements, where $k$ is any number such that $1\leq k<n$, by adding
to $F^{(n)}$ suitable (functions of) elements of ${\cal U}_l$. A
simple possible choice is to put
\begin{equation}\label{f2n-k}
F_{n+i}= \im R_{m^{(i)}}\quad \text{for }\ i=1, \dots,n-k\,,
\end{equation}
where $m^{(1)}, \dots, m^{(n-k)}$ are $n-k$ linearly independent
elements of $\zn^n$, such that $r^{(i)}\cdot m^{(j)}=0\
\forall\,i=1, \dots, k$ and $\forall\,j=1, \dots, n-k$. Using the
second of (\ref{jr}) it is immediate to see that $\{F_i,F_j\}=0$
$\forall\, i=1, \dots, k$ and $\forall\, j=1, \dots, 2n-k$.
Furthermore, it is not difficult to show that the set
$F^{(k)}=(F_1, \dots, F_{2n-k})$ is functionally independent. It
follows that $F^{(k)}$ is an integrable set with $k$ central
elements $(F_1, \dots, F_k)=(J_{r^{(i)}}, \dots, J_{r^{(k)}})$.

It is also possible to construct integrable sets of functions
whose central subset contains elements of $\aop_l$ or $\bop_l$.
For example, one can obtain a set with $k=n$ by taking $F_i=
J_{r^{(i)}}$ for $i=1, \dots, n-1$, $F_n= \im R_m \in \bop_l$,
where $r^{(1)}=l$ and $m\cdot r^{(i)}=0 \ \forall\, i=1, \dots,
n-1$. More generally, let us consider the sets of functions ${\cal
R}_1:= (\im R_{m^{(1)}}, \dots, \im R_{m^{(k')}})$ and ${\cal
R}_2:= (R_{m^{(k'+1)}}, \dots, R_{m^{(h)}})$, where $m^{(1)},
\dots, m^{(k')}, m^{(k'+1)}, \dots, m^{(h)}\in \zn^{n}$ are
linearly independent vectors such that $m^{(i)}\cdot l=0$ and
$m^{(i)}_p m^{(j)}_p =0\ \forall\, i=1, \dots, h$ and $j=1, \dots,
k'$, and $\forall\, p=1, \dots, n$. The latter condition obviously
ensures that $\{\im R_{m^{(i)}}, \im R_{m^{(j)}}\}=0$. For
example, for $n\geq 6$ oscillators of equal frequencies, i.e.,
$l_i=1\ \forall\,i=1, \dots, n$, a pair of sets of such type, with
$k'=3$ and $h=n-4$, is given by ${\cal R}_1=(R_{12}, R_{34},
R_{56})$, ${\cal R}_2=(R_{78}, R_{79}, \dots, R_{7n})$, where the
functions $R_{ij}$ are defined by formula (\ref{rij}). Clearly one
must have in general $k' \leq n/2$. Obviously the case $h=k'$
corresponds to ${\cal R}_2= \emptyset$. It is also easy to see
that, when $h>k'$, then one has necessarily $h < n-k'$.
Furthermore, let us take $n-k'$ linearly independent vectors
$r^{(1)}, \dots, r^{(n-k')}\in \zn^{n}$, with $r^{(1)}=l$,
$r^{(i)} \cdot m^{(j)}=0\ \forall \,i=1, \dots, n-k'$ and $j=1,
\dots, k'$, in such a way that the first $n-h$ vectors of this set
satisfy the additional relations $r^{(i)} \cdot m^{(j)}=0\ \forall
\,i=1, \dots, n-h$ and $j=k'+1, \dots, h$. It is then easy to see
that the set
\begin{equation}\label{r1r2}
(F_1, \dots, F_{2n-k})= ({\cal J}_1, {\cal R}_1; {\cal J}_2, {\cal
R}_2)\,,
\end{equation}
where $k=n-h+k'$, ${\cal J}_1=(J_{r^{(1)}}, \dots, J_{r^{(n-
h)}})$, ${\cal J}_2=(J_{r^{(n-h+1)}}, \dots, J_{r^{(n-k')}})$, is
an integrable set with $k$ central elements $(F_1, \dots, F_{k})=
({\cal J}_1, {\cal R}_1)$.

\subsection{Additional symmetries for equal frequencies}
\label{eqfr} It often occurs in physics that all frequencies, or
part of them, are equal to one another or have the same absolute
value. In these cases there exist also other types of integrable
sets. For example, let us suppose that $l_i= l_j\ \forall \,i,j=1,
\dots, n$. In this case the system with hamiltonian $F_1$ is
invariant with respect to the action of the group $SO(n)$ on
configuration space,
and can be identified with the system describing a point particle
moving in the central potential $U(r)=r^2/2$ in $n$-dimensional
space. It is then possible to construct additional integrable sets
of functions by making use of proposition 
2.4 of \cite{part3}. In the present case, the procedure can also
be generalized to the case in which only the absolute values of
the frequencies are equal to one another. More precisely, let us
suppose that $l_i= \epsilon_i$ for $i=1, \dots, n$, where
$\epsilon_i= \pm 1\ \forall\, i=1, \dots, n$. Consider the
functions
\begin{equation} \label{wi}
w_i:= \begin{cases} z_i &\text{if }\ \epsilon_i=+1\,,  \\
\bar z_i &\text{if }\ \epsilon_i=-1\,, \end{cases}
\end{equation}
for $i=1, \dots, n$, and let $\bar w_i$ denote the complex
conjugate of $w_i$. We have
\[
\{w_i,w_j\}=\{\bar w_i,\bar w_j\}= 0\,, \qquad \{w_i,\bar w_j\}= i
\epsilon_i\delta_{ij}\,,
\]
and $I_i= z_i \bar z_i= w_i \bar w_i \ \forall\, i,j=1, \dots, n$.
It follows that $\{I_i, w_j\}= -i \delta_{ij}\epsilon_j w_j$, and
\begin{equation} \label{hw}
\{F_1, w_j\}= -iw_j \,, \qquad \{F_1, \bar w_j\}= i\bar w_j \,,
\end{equation}
where $F_1= \sum_{i=1}^n \epsilon_i I_i$. Let us consider the
momenta
\begin{equation}\label{pijo}
P_{ij}= i(w_i \bar w_j- w_j \bar w_i)= -P_{ji}\,.
\end{equation}
From (\ref{hw}) it follows that
\begin{equation}\label{fpij}
\{F_1, P_{ij}\}=0 \quad \forall \ i,j=1, \dots, n \,.
\end{equation}
Note that
\[
P_{ij}= \begin{cases} \pm i (z_i \bar z_j- z_j \bar z_i)= \pm(x_i
p_j - x_j p_i) &\text{if }\
\epsilon_i= \epsilon_j= \pm 1\,, \\
\pm i (z_i z_j- \bar z_i \bar z_j)= \mp(x_i p_j + x_j p_i)
&\text{if }\ \epsilon_i= -\epsilon_j= \pm 1\,. \end{cases}
\]
Hence, for $\epsilon_i= \epsilon_j=1$, $P_{ij}$ is formally
identical to the momentum
introduced in \cite{part3} for a particle in a central force field
(although the physical meaning of variables $x$ and $p$ is here
different). We have
\begin{equation} \label{ppoisse}
\{P_{ij}, P_{hk}\} = -\epsilon_i\delta_{ih} P_{jk} - \epsilon_j
\delta_{jk} P_{ih} + \epsilon_i\delta_{ik} P_{jh} + \epsilon_j
\delta_{jh} P_{ik} \,,
\end{equation}
which has to be compared with 
the analogous Poisson bracket in \cite{part3}. It is then easy to
see that
\begin{equation}\label{rotpo}
\{P^2, P_{ij}\}=0 \quad \forall \ i,j=1, \dots, n \,,
\end{equation}
where
\[
P^2:= \sum_{i<j}\epsilon_i \epsilon_j P_{ij}^2 \,.
\]

Since (\ref{rotpo}) is formally identical with the analogous
poisson bracket for a particle in a central force field, it is
easy to see that proposition
2.4 of \cite{part3} can be generalized in the following way.
\begin{prop}\label{bambo}
For any $n\geq 2$, for any sequence $(\epsilon_1, \dots,
\epsilon_n)$, with $\epsilon_i= \pm 1\ \forall\, i=1, \dots, n$,
and for any $z=1, \dots, n-1$, it is possible to construct in a
recursive manner sets $Z_{n,z}$ and $L_{n,z}$ of polynomial
functions of degree $\leq 2$ in the variables $P_n:=(P_{ij}, 1\leq
i<j \leq n)$, with the following properties:
\begin{enumerate}
\item $Z_{n,z}$ contains $z$ elements,

\item $L_{n,z}$ contains $2(n-z-1)$ elements,

\item the set $\Pi_{n,z}:= (Z_{n,z}, L_{n,z})$ is functionally
independent,

\item $\{Z_{n,z}, \Pi_{n,z}\}=0$.
\end{enumerate}
\end{prop}
In this proposition, momenta $P_{ij}$ are obviously defined
according to formula (\ref{pijo}). From (\ref{fpij}) it follows
that $\{F_1, \Pi_{n,z}\}=0$. Hence the sets of functions
$F_{n,z}:= (F_1, Z_{n,z}; L_{n,z})$ are integrable sets, with
subset of central elements $(F_1, Z_{n,z})$. We have $\sharp
F_{n,z}= 2n-k$ and $k=z+1$. Hence the number $k$ of central
elements of these sets can take all values from $k=2$ to $k=n$.

We are now able to describe a general class of integrable sets,
which includes those presented above as particular cases, and
which can be applied to any arbitrary set of frequencies $l\in
(\zn \setminus \{0\})^n$. Let us divide the frequencies into $u$
groups, with $1\leq u\leq n$, so that each group contains one or
more frequencies whose absolute values are equal to one another.
Namely, without loss of generality, we suppose that
\begin{equation}\label{ls}
l_i = \begin{cases}\epsilon_i s_1 & \quad \text{for }\ i=1, \dots, p_1 \,, \\
\epsilon_i s_2 & \quad \text{for }\ i=p_1+ 1, \dots, p_2 \,, \\
\ \ \vdots & \quad\ \,\vdots \\
\epsilon_i s_u & \quad \text{for }\ i=p_{u-1}+1, \dots, n \,,
\end{cases}
\end{equation}
where $s_1, \dots, s_u \in \nn$, $\epsilon_i= \pm1\ \forall\, i=1,
\dots, n$, $1\leq p_1 <p_2 <\dots$ $<p_{u-1}< n$. The number of
frequencies contained in the various groups are $q_h= p_{h}-
p_{h-1} \geq 1$ for $h=1, \dots, u$, where $p_0:=0$, $p_u:=n$.
Obviously $\sum_{h=1}^u q_h=n$. If the absolute values of all
frequencies of the system are pairwise different, one can only
take $u=n$, $q_h=1\ \forall\,h=1, \dots, n$. Note however that, at
variance with the analogous partition of the generalized moments
of inertia $\lambda_1, \dots, \lambda_n$ of a rigid body, given
in \cite{part3}, here we do not require that the absolute values
of the frequencies of different groups be pairwise different.
Hence, it is possible in general that $s_h= s_j$ for some $h\neq
j$. This means that, if the system includes any set of more than
one oscillators, whose frequencies have a common absolute value,
one is free to divide them into subgroups in any arbitrary
possible way. Each choice will lead to different realizations of
integrable sets of functions, according to the procedure which we
are going to describe.

Let $w_i$ be defined by formula (\ref{wi}) for $i=1, \dots, n$.
For each $h=1, \dots, u$, let us consider the sets of momenta
\begin{equation}\label{ph}
P_{(h)}=\begin{cases}\emptyset & \quad \text{if }\ q_h=1\,,\\
(P_{ij}, \ p_{h-1}< i<j\leq p_h) & \quad \text{if }\ q_h>1\,,
\end{cases}
\end{equation}
where $P_{ij}$ is defined by formula (\ref{pijo}). Note that
$\{P_{(h)}, P_{(j)}\}=0$ if $h\neq j$. For each $h$ such that
$q_h>1$, let us construct sets of functions $Z_{q_h, z_h}
(P_{(h)})$ and $L_{q_h, z_h} (P_{(h)})$ by means of proposition
\ref{bambo}, where $z_h$ can be arbitrarily chosen among the
possible values $1,2, \dots, q_h-1$. For those $h$ such that
$q_h=1$ we put instead $z_h=0$ and $Z_{q_h, z_h} (P_{(h)})=
L_{q_h, z_h} (P_{(h)})=\emptyset$. We have
\begin{align*}
\sharp Z_{q_h, z_h} (P_{(h)}) &= z_h \,, \\
\sharp L_{q_h, z_h} (P_{(h)}) &= 2(q_h-z_h-1)\,,
\end{align*}
and
\begin{equation}\label{zpi}
\{Z_{q_h, z_h} (P_{(h)}), \Pi_{q_j, z_j} (P_{(j)})\} =0
\end{equation}
$\forall\, h,j=1, \dots, u$, where $\Pi_{q_j, z_j}= (Z_{q_h, z_h},
L_{q_h, z_h})$.

We can write $F_1= \sum_{h=1}^u s_h K_h$, where
\[
K_h:= \sum_{i=p_{h-1}+1}^{p_h} \epsilon_i I_i\,.
\]
In analogy with (\ref{fpij}) we have
\begin{equation}\label{fpijs}
\{K_h, P_{(j)}\}=0\,,
\end{equation}
so that
\begin{equation}\label{kpij}
\{K_h, \Pi_{q_j, z_j} (P_{(j)})\}=0 \quad \forall\, h,j=1, \dots,
u\,.
\end{equation}
This implies, in particular, that $\{F_1, \Pi_{q_j, z_j}
(P_{(j)})\}=0$. Let us introduce the functions
\[
W_h:= \sum_{i=p_{h-1}+1}^{p_h} \epsilon_i w_i^2\,, \qquad h=1,
\dots,u\,.
\]
It is easy to check that
\begin{align}
\{W_h, W_j\} &=0\,, \qquad \{W_h, \bar W_j\} =4i \delta_{jh} K_h
\,, \label{ww} \\
\{W_h, P_{(j)}\} &=0\,, \qquad \{W_h, K_j\} =2i \delta_{jh} W_h
\label{fw2}
\end{align}
$\forall\, j,h=1, \dots, u$, where $\bar W_h$ denotes as usual the
complex conjugate of $W_h$. Note also the relation
\begin{equation} \label{wkp}
W_h \bar W_h= K_h^2- P_{(h)}^2 \,,
\end{equation}
where
\[
P_{(h)}^2:= \sum_{p_{h-1}< i<j\leq p_h}\epsilon_i \epsilon_j
P_{ij}^2 \,.
\]

We are now going to construct integrable sets of functions which
contain the set $\tilde F= \left(F_1, \Pi_{q_1, z_1} (P_{(1)}),
\dots, \Pi_{q_u, z_u} (P_{(u)})\right)$. For the choice of the
additional elements, we shall follow a procedure which is similar
to the one we used before in section \ref{gs}. Here the sets of
functions $(K_1, \dots, K_u)$ and $(W_1, \dots, W_u)$ will play
the role that was formerly played by the sets $(I_1, \dots, I_n)$
and $(z_1, \dots, z_n)$ respectively. We shall in fact replace
formula (\ref{rm}) by
\begin{equation} \label{rm2}
R_m= \Omega_1^{m_1} \Omega_2^{m_2}\cdots \Omega_u^{m_u} \,,
\end{equation}
where $m\in \zn^u$ and
\begin{equation} \label{zeta2}
\Omega_h^{m_h}:= \begin{cases} W_h^{m_h} &\text{if }\ m_h\geq 0\,,  \\
\bar W_h^{-m_h} &\text{if }\ m_h< 0\,. \end{cases}
\end{equation}
Moreover, we shall replace (\ref{j}) by
\begin{equation} \label{j2}
J_r= \sum_{h=1}^u r_h K_h\,,
\end{equation}
where $r\in \zn^u$. Note that $F_1= J_s$. Using the second of
(\ref{fw2}) we easily obtain the Poisson brackets
\begin{equation} \label{jr2}
\{J_r, \re R_m\}= 2r\cdot m\, \im R_m \,,\qquad \{J_r, \im R_m\}=
-2r\cdot m \,\re R_m \,,
\end{equation}
which have to be compared with (\ref{jr}).

Consider the sets of functions ${\cal A}_s= (\re R_m,\ s\cdot
m=0)$, ${\cal B}_s= (\im R_m, \ s\cdot m=0)$, $K=(K_1, \dots,
K_u)$ and $P^2= (P^2_{(1)}, \dots, P^2_{(u)})$. Using relations
(\ref{kpij})--(\ref{wkp}), it is easy to verify that ${\cal U}_s:=
(K, P^2, {\cal A}_s, {\cal B}_s)$ is a set of functions in
involution with $\tilde F$, which is closed with respect to
Poisson brackets. We can obtain an integrable set analogous to the
one defined by formulas (\ref{f1n})--(\ref{f2n-k}), by adding to
the set $\tilde F$ suitable combinations of elements of ${\cal
U}_s$. Let $r^{(1)}, \dots, r^{(u)}$ be $u$ linearly independent
elements of $\zn^u$, with $r^{(1)}=s$. Choose $k'\in \nn$ such
that $1\leq k'\leq u$, and construct the set of $u-k'$ functions
\[
{\cal R}=(\im R_{m^{(1)}}, \dots, \im R_{m^{(u-k')}}) \,,
\]
where $m^{(1)}, \dots, m^{(u-k')}$ are $u-k'$ linearly independent
elements of $\zn^u$, such that $r^{(h)}\cdot m^{(j)}=0\
\forall\,h=1, \dots, k'$ and $\forall\,j=1, \dots, u-k'$.
Obviously ${\cal R}=\emptyset$ if $k'=u$. Consider then the set of
functions
\begin{equation} \label{fset}
F=({\cal J}_1, {\cal Z}; {\cal J}_2, {\cal R}, {\cal L})\,,
\end{equation}
where
\begin{align*}
{\cal J}_1&= (J_{r^{(1)}}, \dots, J_{r^{(k')}}) \,, \\
{\cal J}_2&= (J_{r^{(k'+1)}}, \dots, J_{r^{(u)}}) \,, \\
{\cal Z}&= \left(Z_{q_1, z_1} (P_{(1)}), \dots, Z_{q_u, z_u}
(P_{(u)})\right) \,,\\
{\cal L}&= \left(L_{q_1, z_1} (P_{(1)}), \dots, L_{q_u, z_u}
(P_{(u)}) \right) \,.
\end{align*}
Note that $J_{r^{(1)}}= J_{s}= F_1$. We have
\[
\sharp {\cal J}_1= k' \,, \qquad \sharp {\cal J}_2= u-k' \,,
\qquad \sharp {\cal R}= u-k' \,,
\]
\[
\sharp {\cal Z}= z:= \sum_{h=1}^u z_h \,, \qquad \sharp {\cal L}=
2\sum_{h=1}^u (q_h-z_h-1) =2(n-z-u)\,.
\]
Hence
\begin{align*}
&\sharp({\cal J}_1, {\cal Z})=k :=k'+z\,, \\
&\sharp F= 2n- z- k'= 2n-k\,.
\end{align*}

We have obviously
\[
\{{\cal J}_1, {\cal J}_1\}=\{{\cal J}_1, {\cal J}_2\}=0\,.
\]
From (\ref{kpij}) it follows that
\[
\{{\cal J}_1, {\cal Z}\}=\{{\cal J}_1, {\cal L}\}=\{{\cal J}_2,
{\cal Z}\}=0\,.
\]
From the second of (\ref{jr2}) it follows that
\[
\{{\cal J}_1, {\cal R}\}=0\,.
\]
From (\ref{zpi}) it follows that
\[
\{{\cal Z}, {\cal Z}\}=\{{\cal Z}, {\cal L}\}=0 \,.
\]
Finally, from the first of (\ref{fw2}) it follows that
\[
\{{\cal Z}, {\cal R}\}=0 \,.
\]
It is also possible to show that the set $F$ is functionally
independent. We are therefore able to formulate the following
\begin{prop}\label{osc}
The set $F$ defined by (\ref{fset}) is an integrable set of
functions with $k$ central elements $({\cal J}_1, {\cal Z})$,
where $k=k'+ z$.
\end{prop}

From this propositions it obviously follows that any system with
hamiltonian $H=f({\cal J}_1, {\cal Z})$, where $f$ is any function
of $k$ variables, is integrable with the same integrable set $F$.

It is also possible to obtain an integrable set of functions
$F\supset \tilde F$, whose central subset includes elements of
$({\cal A}_s, {\cal B}_s)$. To this purpose, we shall follow a
procedure similar to the one which led us before to the set
(\ref{r1r2}). If $1\leq k'\leq u/2$, let $m^{(1)}, \dots,
m^{(k')}, m^{(k'+1)}, \dots, m^{(h)}\in \zn^{u}$ be linearly
independent vectors, such that $m^{(i)}\cdot s=0$ and $m^{(i)}_p
m^{(j)}_p =0\ \forall\, i=1, \dots, h$ and $j=1, \dots, k'$, and
$\forall\, p=1, \dots, n$. It is easy to see that one must have
either $h=k'$ or $ k'<h< u-k'$. Furthermore, let us take $u-k'$
linearly independent vectors $r^{(1)}, \dots, r^{(u-k')}\in
\zn^{u}$, with $r^{(1)}=s$, $r^{(i)} \cdot m^{(j)}=0\ \forall
\,i=1, \dots, u-k'$ and $j=1, \dots, k'$, in such a way that the
first $u-h$ vectors of this set satisfy the additional relations
$r^{(i)} \cdot m^{(j)}=0\ \forall \,i=1, \dots, u-h$ and $j=k'+1,
\dots, h$. Let us then consider the set
\begin{equation}\label{fset2}
F= ({\cal J}_1, {\cal R}_1, {\cal Z}; {\cal J}_2, {\cal R}_2,
{\cal L})\,,
\end{equation}
where
\begin{align*}
{\cal J}_1 &=(J_{r^{(1)}}, \dots, J_{r^{(u- h)}})\,, \\
{\cal J}_2 &=(J_{r^{(n-h+1)}}, \dots, J_{r^{(u-k')}})\,, \\
{\cal R}_1 &= (\im R_{m^{(1)}}, \dots, \im R_{m^{(k')}})\,,
\\
{\cal R}_2 &= (\im R_{m^{(k'+1)}}, \dots, \im R_{m^{(h)}})\,.
\end{align*}
It is easy to see that $F$ contains $2n-k$ elements, where
$k=u-h+k'+z$.
\begin{prop}\label{osc2}
The set $F$ defined by (\ref{fset2}) is an integrable set of
functions with $k$ central elements $(F_1, \dots, F_{k})= ({\cal
J}_1, {\cal R}_1, {\cal Z})$.
\end{prop}

Hence any system with hamiltonian $H=f({\cal J}_1, {\cal R}_1,
{\cal Z})$, where $f$ is any function of $k$ variables, is
integrable with the same integrable set $F$.

\section{Quantum oscillators with completely resonant set of
frequencies} \label{qosc}

In order to maintain the correspondence between our notation and
the one usually employed in physics, in this section we shall
associate with the classical impulses $p$ the complex operators
$\hat p= -i\partial/\partial x$, where $i=\sqrt{-1}$. We have with
this convention $\hat p=\hat p^*$, where $\hat p^*$ denotes the
hermitian conjugate of the operator $\hat p$. It is obvious that
this modification does not substantially affect the general
results on quantization which have been obtained in the preceding
papers of this series, although some formulas have to be corrected
by the introduction of one or more factors $i$. The standard
canonical commutation relations become
\begin{equation} \label{ccr2}
[x_i,x_j]= 0\,, \qquad [\hat p_i,\hat p_j]= 0\,, \qquad [\hat p_i,
x_j] = -i\delta_{ij}
\end{equation}
for $i,j= 1, \dots, n$, where $\delta_{ij}$ is the Kr\"onecker
symbol. The standard quantization of the functions $z_j, \bar z_j$
is given by
\[
\hat z_j=\frac{x_j +i\hat p_j}{\sqrt 2} \,, \qquad \hat z^*_j=
\frac{x_j -i\hat p_j}{\sqrt 2} \,,
\]
where $\hat z^*_j$ denotes the hermitian conjugate of the operator
$\hat z_j$. These operators satisfy the commutation relations
\begin{equation}
[\hat z_i,\hat z_j]=[\hat z^*_i,\hat z^*_j]= 0\,, \qquad [\hat
z_i,\hat z^*_j]= \delta_{ij}
\end{equation}
for $i,j=1,\dots,n$. Note that, with the conventions presently
adopted for the quantization of impulses $p$, the Poisson bracket
of two functions is now replaced by the commutator of the
corresponding operators multiplied by $i$.

Let as above $F_1= l_1 I_1 + \cdots + l_n I_n$ be the hamiltonian
function of the system describing linear oscillations with a
completely resonant set of frequencies. The standard quantization
$\hat F_1$ of $F_1$ clearly coincides with its symmetrization
$F_1^{\rm sym}$ with respect to $(x, \hat p)$. On the other hand,
since the operators $\hat z, \hat z^*$ are linear combinations of
the operators $x, \hat p$, it is easy to see that the
symmetrization with respect to $(x, \hat p)$ is equivalent to the
symmetrization with respect to $(\hat z, \hat z^*)$. We thus have
$\hat F_1= l_1 \hat I_1 + \cdots + l_n \hat I_n$, where
\[
\hat I_j= I_j^{\rm sym}= \frac{x_j^2 +\hat p_j^2}2 =\frac{\hat
z^*_j\hat z_j +\hat z_j\hat z^*_j}2 \,.
\]

In the general situation described by relations (\ref{ls}), we can
write $\hat F_1= s_1 \hat K_1 + \cdots + s_u \hat K_u$, where
\[
\hat K_h:= \sum_{i=p_{h-1}+1}^{p_h} \epsilon_i \hat I_i\,, \qquad
h=1, \dots, u\,.
\]
Let us introduce the operators
\begin{equation*} 
\hat w_i:= \begin{cases} \hat z_i &\text{if }\ \epsilon_i=+1\,,  \\
\hat z_i^* &\text{if }\ \epsilon_i=-1\,, \end{cases}
\end{equation*}
which satisfy the relations
\begin{equation}
[\hat w_i,\hat w_j]=[\hat w^*_i,\hat w^*_j]= 0\,, \qquad [\hat
w_i,\hat w^*_j]= \epsilon_i\delta_{ij}\,,
\end{equation}
\[
\hat I_j= \frac{\hat w^*_j\hat w_j +\hat w_j\hat w^*_j}2 \,,
\]
for $i,j=1,\dots,n$. We then define
\begin{align*}
\hat P_{ij}&= i(\hat w_i \hat w_j^*- \hat w_j \hat w_i^*)= -\hat
P_{ji}\,, \qquad i,j=1, \dots, n \,,\\
\hat W_h&= \sum_{i=p_{h-1}+1}^{p_h} \epsilon_i \hat w_i^2\,,
\qquad h=1, \dots,u\,.
\end{align*}
Clearly all operators of the set $(\hat K, \hat P, \hat W, \hat
W^*)$ coincide with the symmetrization with respect to $(\hat w,
\hat w^*)$ of the corresponding classical functions. It then
follows from proposition 
3.1 (case 1) of \cite{part2} that the commutators (multiplied by
$i$) between these operators have the same form as the Poisson
brackets between the corresponding classical functions. In
particular, from (\ref{fpijs}), (\ref{ww}) and (\ref{fw2}) we get
\begin{equation}\label{kp}
[\hat K_h, \hat P_{(j)}]=0\,, \qquad [\hat W_h, \hat P_{(j)}]
=0\,,
\end{equation}
\begin{equation}\label{wwq}
[\hat W_h, \hat W_j] =0\,, \qquad [\hat W_h, \hat W_j^*] =4
\delta_{jh} \hat K_h \,, \qquad [\hat W_h, \hat K_j] =2
\delta_{jh} \hat W_h
\end{equation}
$\forall\, j,h=1, \dots, u$, where $\hat P_{(j)}$ is the standard
quantization of the set $P_{(j)}$ defined by formula (\ref{ph}).
We also introduce the operators
\begin{equation} \label{rm2q}
\hat R_m= \hat \Omega_1^{m_1} \hat \Omega_2^{m_2}\cdots \hat
\Omega_u^{m_u} \,,
\end{equation}
where $m\in \zn^u$ and
\begin{equation} \label{zeta2q}
\hat \Omega_h^{m_h}:= \begin{cases} \hat W_h^{m_h} &\text{if }\ m_h\geq 0\,,  \\
(\hat W_h^*)^{-m_h} &\text{if }\ m_h< 0\,. \end{cases}
\end{equation}
Note that one need not specify the ordering of the operators on
the right-hand sice of (\ref{rm2q}), since all these operators
commute according to (\ref{wwq}).

Let us consider the integrable set (\ref{fset}) for the classical
system. We put
\begin{align*}
\hat{\cal J}_1&= (\hat J_{r^{(1)}}, \dots, \hat J_{r^{(k')}}) \,, \\
\hat {\cal J}_2&= (\hat J_{r^{(k'+1)}}, \dots, \hat J_{r^{(u)}})
\,,
\end{align*}
where
\begin{equation*} 
\hat J_{r^{(i)}}= \sum_{h=1}^u r^{(i)}_h \hat K_h\,, \qquad i=1,
\dots, u \,.
\end{equation*}
We put also
\begin{align*}
\hat {\cal Z}&= (\hat Z_{q_1, z_1} (\hat P_{(1)}), \dots, \hat
Z_{q_u, z_u} (\hat P_{(u)})) \,,\\
\hat {\cal L}&= (\hat L_{q_1, z_1} (\hat P_{(1)}), \dots, \hat
L_{q_u, z_u} (\hat P_{(u)})) \,,
\end{align*}
where polynomials $\hat Z_{n,z}$ and $\hat L_{n,z}$ have the same
form as the classical ones $Z_{n,z}$ and $L_{n,z}$ (symmetrization
is here unnecessary, since in these polynomials all monomials of
degree 2 are squares of components of $\hat P$). We finally put
\[
\hat{\cal R}=(\im \hat R_{m^{(1)}}, \dots, \im \hat
R_{m^{(u-k')}}) \,,
\]
where we define the ``imaginary part'' of the operator $\hat
R_{m^{(h)}}$ as
\[
\im \hat R_{m^{(h)}}:= \frac 1{2i} \left(\hat R_{m^{(h)}}- \hat
R_{m^{(h)}}^*\right) \,.
\]
The operators of the set $\hat{\cal R}$ obviously coincide with
the symmetrization with respect to $(\hat W, \hat W^*)$ of the
functions of the set ${\cal R}$. From the isomorphism between the
two Lie algebras respectively generated by the functions $(K, P,
W, \bar W)$ and by the operators $(\hat K, \hat P, \hat W, \hat
W^*)$, it follows that one can deduce the commutation relations
\[
[\hat{\cal J}_1, \hat{\cal J}_1]=[\hat{\cal J}_1, \hat{\cal
J}_2]=0\,,
\]
\[
[\hat{\cal J}_1, \hat{\cal Z}]=[\hat{\cal J}_1, \hat{\cal L}]
=[\hat{\cal J}_2, \hat{\cal Z}]= [\hat{\cal J}_1, \hat{\cal R}]=0
\]
from the corresponding Poisson bracket relations, by making use of
proposition 
4.2, case a, of \cite{part2}. In order to prove that
\[
[\hat{\cal Z}, \hat{\cal Z}]=[\hat{\cal Z}, \hat{\cal L}]=0\,,
\]
we note that relations
\[
[\hat Z_{q_h, z_h} (\hat P_{(h)}), \hat Z_{q_j, z_j} (\hat
P_{(j)})] = [\hat Z_{q_h, z_h} (\hat P_{(h)}), \hat L_{q_j, z_j}
(\hat P_{(j)})]=0
\]
for $h\neq j$ follow from $[\hat P_{(h)}, \hat P_{(j)}]=0$. For
$h=j$ they can instead be deduced from proposition \ref{bambo}, by
repeating the same arguments that were used in \cite{part3} to
deduce proposition 2.7 
from proposition 2.4 in the analogous situation of the central
force field. Finally, from the second of (\ref{kp}) it follows
that
\[
[\hat{\cal Z}, \hat{\cal R}]=0 \,.
\]

\begin{prop}\label{oscq}
The set $\hat F=(\hat{\cal J}_1, \hat{\cal Z}; \hat{\cal J}_2,
\hat{\cal R}, \hat{\cal L})$ is a quasi-integrable set of
operators with $k$ central elements $(\hat{\cal J}_1, \hat{\cal
Z})$, where $k=k'+ z$.
\end{prop}
\begin{proof}
We have just shown that the elements of the set $\hat F$ satisfy
the required commutation relations. It is not too difficult to
complete the proof of this proposition, by showing that the set is
also quasi-independent.
\end{proof}
From  this proposition it follows that any quantum system with
hamiltonian operator $\hat H=f(\hat{\cal J}_1, \hat{\cal Z})$,
where $f$ is an arbitrary polynomial of $k$ variables, is
integrable with the same integrable set of operators $\hat F$.

In a similar way one can quantize the integrable set
(\ref{fset2}). Let us define in this case
\begin{align*}
\hat{\cal J}_1 &=(\hat J_{r^{(1)}}, \dots, \hat J_{r^{(u- h)}})\,, \\
\hat{\cal J}_2 &=(\hat J_{r^{(n-h+1)}}, \dots, \hat J_{r^{(u-k')}})\,, \\
\hat{\cal R}_1 &= (\im \hat R_{m^{(1)}}, \dots, \im \hat
R_{m^{(k')}})\,, \\
\hat{\cal R}_2 &= (\im \hat R_{m^{(k'+1)}}, \dots, \im \hat
R_{m^{(h)}})\,.
\end{align*}
\begin{prop}
The set $\hat F=(\hat{\cal J}_1, \hat{\cal R}_1, \hat{\cal Z};
\hat{\cal J}_2, \hat{\cal R}_2, \hat{\cal L})$ is a
quasi-integrable set of operators with $k$ central elements
$(\hat{\cal J}_1, \hat{\cal R}_1, \hat{\cal Z})$, where
$k=u-h+k'+z$.
\end{prop}
Hence any quantum system with hamiltonian operator $\hat
H=f(\hat{\cal J}_1, \hat{\cal R}_1, \hat{\cal Z})$, where $f$ is
an arbitrary polynomial of $k$ variables, is integrable with the
same integrable set of operators $\hat F$.

\section{The case $n=3$}

In section \ref{closc} we have constructed a remarkable class of
integrable sets of functions for an arbitrary system of resonant
oscillators. However there may exist other integrable sets which
do not belong to the class we have considered. In the present
section we shall make an attempt to systematically classify all
integrable sets of polynomial functions in the canonical variables
$(x,p)$, or equivalently $(z, \bar z)$, for $n=3$. Since our goal
is to obtain the largest possible class of integrable systems, we
shall restrict our attention to integrable sets of functions
having the largest number $k$ of central elements, i.e., $k=3$.
Hence, we shall aim at characterizing all sets $F=(F_1, F_2, F_3)$
of functionally independent polynomial functions, such that $F_1 =
J_l= l_1 I_1+ l_2 I_2+ l_3 I_3 $ and $\{F_i, F_j\}=0$ for
$i,j=1,2,3$. To each of these sets there corresponds a class of
integrable systems with hamiltonian $H=f(F)$, where $f$ is an
arbitrary function of three variables. This class is obviously the
same for two integrable sets which are functionally equivalent to
each other, i.e., such that the elements of one set are locally
functions of the elements of the other set. Therefore in the
following, whenever we shall mention an integrable set, we shall
implicitly refer to the whole equivalence class to which it
belongs.

If the three frequencies are pairwise different, i.e., $l_i\neq
l_j$ for $i\neq j$, according to the discussion of section
\ref{gs} there exist integrable sets with $k=3$ of the form
\begin{equation}\label{fint1}
F=\big(F_1, J_{r}, f(I_1, I_2, I_3, \im R_{m}) \big)\,,
\end{equation}
where $f$ is an arbitrary function of four variables such that the
set $F$ is functionally independent. Here $r\in \zn^3$ and $m \in
\zn^3$ are two nonvanishing vectors such that $r$ is linearly
independent of $l$, and $l \cdot m = r \cdot m=0$. For example,
the trivial integrable set $F=(F_1, I_2, I_3)$ corresponds to
taking $r= (0,1,0)$ and $f(x_1, x_2, x_3, x_4) =x_3$ in
(\ref{fint1}). Note that there exist only 4 functionally
independent functions in involution with both $F_1$ and $F_2=
J_{r}$. Since 4 such functions are $(I_1, I_2, I_3, \im R_{m})$,
for any integrable set such that $F_2= J_{r}$ we can locally write
$F_3= f(I_1, I_2, I_3, \im R_{m})$ as in (\ref{fint1}). For
example, owing to (\ref{imre}), we obtain $F_3= \re R_{m}$ by
taking $f(x_1, x_2, x_3, x_4)= \pm \sqrt{x_1^{|m_1|} x_2^{|m_2|}
x_3^{|m_3|}- x_4^2}$.

For the classification of integrable sets it is useful to
introduce, besides the concept of functional equivalence
given in \cite{part3}, also that of ``canonical equivalence''
between integrable sets.
\begin{defn}
We say that two integrable sets $F$ and $F'$ are {\it canonically
equivalent} if there exists a symplectic (i.e., linear and
canonical) transformation which transforms one set into the other.
In such a case, we say that also the two classes of functional
equivalence, to which the two sets respectively belong, are
canonically equivalent. This means that, if $G$ is functionally
equivalent to $F$, and $G'$ is functionally equivalent to $F'$,
then we say that $G$ is canonically equivalent to $G'$. In
particular, two functionally equivalent sets are also canonically
equivalent.
\end{defn}
Let ${\cal G}_{l}$ be the group of all the symplectic
transformations $g:(z, \bar z)\mapsto (Z, \bar Z)$ which leave
$F_1$ invariant, i.e., such that $F_1(z,\bar z)= F_1(Z,\bar Z)$,
or explicitly
\begin{equation}
l_1 z_1 \bar z_1 +l_2 z_2 \bar z_2+ l_3 z_3 \bar z_3 = l_1 Z_1
\bar Z_1 +l_2 Z_2 \bar Z_2+ l_3 Z_3 \bar Z_3\ .
\end{equation}
Given an integrable set $F= \left(F_1(z,\bar z), F_2(z,\bar z),
F_3(z,\bar z) \right)$, we can for any $g\in {\cal G}_{l}$
construct a set $F' =\left(F_1(Z,\bar Z), F_2(Z,\bar Z),
F_3(Z,\bar Z)\right)$ which is canonically equivalent to $F$.


Let us consider the Lie algebra ${\cal L}_{l}$ of the second
degree real polynomials which are in involution with $F_1$. The
group ${\cal G}_{l}$ contains all the one-parameter subgroups of
canonical transformations generated by the elements of ${\cal
L}_{l}$. For any $G\in {\cal L}_{l}$ we write the associated
subgroup in the form $ z\mapsto Z(\tau)$, with
\begin{equation}\label{sub}
\frac{d Z(\tau)}{d\tau} =\{Z(\tau),G\}\,, \qquad Z(0)= z\ .
\end{equation}

If the three frequencies are pairwise different, a linear basis of
this algebra is given by the set $(I_1, I_2, I_3)$. Let us
consider the subgroup generated by $G=I_1$. According to
(\ref{sub}) we have
\[
Z_1(\tau)= e^{i\tau} z_1\,, \qquad Z_2(\tau)= z_2\,, \qquad
Z_3(\tau)= z_3\,.
\]
Similar formulas hold for $G=I_2$ and $G=I_2$. Therefore, all
transformations of the group ${\cal G}_{l}$ have the form
\[
Z_1= \exp(i\phi_1) z_1\,, \qquad Z_2= \exp(i\phi_2) z_2\,, \qquad
Z_3= \exp(i\phi_3) z_3\,,
\]
with $\phi_i \in \rn$ for $i=1,2,3$. This means that ${\cal
G}_{l}$ is an abelian group isomorphic to $U(1)\times U(1) \times
U(1)$.

It can be useful to classify integrable sets $F=(F_1, F_2, F_3)$
of polynomial functions according to the degree of polynomials
$F_2$ and $F_3$. We shall always suppose that these polynomials
have been chosen in such a way that $\deg F_2\leq \deg F_3$, and
that there do not exist functionally equivalent sets of lower
degree. More precisely, we suppose that there does not exist
another functionally equivalent integrable set of polynomials
$(F_1, F_2', F_3')$ such that $\deg F_2'<\deg F_2$, or $\deg
F_2'=\deg F_3$ and $\deg F_3'<\deg F_3$.
\begin{defn}
Let $F=(F_1, F_2, F_3)$ be an integrable set, where $F_1= l_1 I_1+
l_2 I_2+ l_3 I_3 $ and $\deg F_2\leq \deg F_3$. We say that $F$ is
a {\it simple integrable set} if $\deg F_2 =2$. We say that a
simple integrable set $F$ has degree $d$ if $\deg F_3 =d$.
\end{defn}
Note that all integrable sets of the form (\ref{fint1}), which
were obtained using the general methods of sections \ref{gs} and
\ref{eqfr}, are simple. Since a symplectic transformation does not
alter the degree of a polynomial function, an integrable set which
is canonically equivalent to a simple integrable set is also
simple, and two canonically equivalent simple integrable sets have
the same degree.

\subsection{The case $|l_1|=|l_2|\neq |l_3|$}\label{case}

If $l_1=l_2$, by applying the results of section \ref{eqfr} we
obtain another type of simple integrable sets, in addition to that
given by formula (\ref{fint1}). It has the form
\begin{equation}\label{fint2}
F=\big(F_1, P_{12}, f(I_1+ I_2, I_3, \im R, \re R)\big)\,,
\end{equation}
where $f$ is an arbitrary function of 4 variables such that the
set $F$ is functionally independent, $P_{12}= 2\im (\bar z_1
z_2)$, and
\begin{equation} \label{r3}
R= \begin{cases}(\bar z_1^2+ \bar z_2^2)^{|l_3|} \,
z_3^{2|l_1|}\qquad &\text{if }l_1 l_3>0\,, \\
(z_1^2+ z_2^2)^{|l_3|} \,z_3^{2|l_1|}\qquad &\text{if } l_1
l_3<0\,.\end{cases}
\end{equation}
In a similar way, one can see that there exist also simple
integrable sets of the form
\begin{equation}\label{fint3}
F=\big(F_1, Q_{12}, f(I_1+ I_2, I_3, \im S, \re S)\big)\,,
\end{equation}
where $Q_{12}= 2\re (\bar z_1 z_2)$, and
\begin{equation} \label{s3}
S= \begin{cases}(\bar z_1^2- \bar z_2^2)^{|l_3|} \,
z_3^{2|l_1|}\qquad &\text{if }l_1 l_3>0\,, \\
(z_1^2- z_2^2)^{|l_3|} \,z_3^{2|l_1|}\qquad &\text{if } l_1
l_3<0\,.\end{cases}
\end{equation}
Analogous integrable sets also exist for $l_1=-l_2$. For the sake
of simplicity, in the following we shall write down explicitly
only the formulas which are valid for $l_1=l_2$.

\begin{prop} \label{ll}
For $l_1= l_2\neq l_3$, the Lie algebra ${\cal L}_{l}$ of the
second degree real polynomials in involution with $F_1$ has linear
dimension five. A basis of this algebra is given by the set
$L=(L_1, \dots, L_5)$, where
\begin{equation} \label{2gcom}\begin{split}
L_1&=-{\rm Re}\{\bar z_1 z_2\}= -\frac {\bar z_1 z_2+z_1\bar
z_2} 2 = -\frac {x_1x_2+p_1p_2} 2 = -\frac 12 Q_{12} \,,\\
L_2&=-{\rm Im}\{\bar z_1 z_2\}= i\frac {\bar z_1 z_2 -z_1\bar
z_2}{2}= \frac {p_1x_2 -x_1p_2} 2 = -\frac 12 P_{12} \,,\\
L_3&= -\frac {I_1- I_2} 2 = \frac {-\bar z_1 z_1+ \bar z_2
z_2} 2 = \frac {-x_1^2- p_1^2 +x_2^2+ p_2^2} 4 \,,\\
L_4&= -\frac {I_1+I_2} 2 = -\frac {\bar z_1 z_1+ \bar z_2
z_2} 2 = -\frac {x_1^2+p_1^2 +x_2^2+p_2^2} 4 \,,\\
L_5 &= I_3 =  \bar z_3 z_3 = \frac {x_3^2+p_3^2} 2 \, .
\end{split}
\end{equation}
The first four of these functions do not depend on $(x_3, p_3)$.
They satisfy the algebraic relation
\begin{equation}
L_1^2+L_2^2 = |\bar z_1 z_2|^2 =I_1I_2 = L_4^2 - L_3^2\ ,
\label{constr}
\end{equation}
or
\begin{equation}
\sum_{i=1}^3 L_i^2 = L_4^2\ . \label{constr1}
\end{equation}
The Poisson brackets between the elements of $L$ are
\begin{align}
\{L_i,L_j\} &=\sum_{k=1}^3\varepsilon_{ijk}L_k \qquad \quad
\text{for } i,j,k=1,2,3 \,,\label{pb1} \\
\{L_\mu,L_4\} &=\{L_\mu,L_5\}=0 \qquad \text{for }\mu=1, \dots,5
\,,\label{pb2}
\end{align}
where $\varepsilon_{ijk}$ is the completely antisymmetric tensor
such that $\varepsilon_{123}=1$.
\end{prop}
\begin{proof}
By using (\ref{f1pab}) it is easy to check that all elements of
${\cal L}_{l}$ can be expressed as linear combinations of elements
of $L=(L_1, \dots, L_5)$. Since the set $L$ is obviously linearly
independent, it forms a basis of ${\cal L}_{l}$. Introducing the
Pauli matrices
\[
\sigma_1=\begin{pmatrix} 0& 1 \\
1&0 \end{pmatrix}\,, \qquad \sigma_2= \begin{pmatrix} 0& -i \\
i&0 \end{pmatrix}\,, \qquad \sigma_3=\begin{pmatrix} 1& 0 \\
0&1 \end{pmatrix}\,,
\]
and defining
\[
\sigma_4= E= \begin{pmatrix} 1& 0 \\
0&1 \end{pmatrix}\,,
\]
we can write
\begin{equation}\label{pauli}
L_\mu= -\frac 12 y^*\sigma_\mu y \,, \qquad \mu=1,2,3,4\,,
\end{equation}
where
\[
y= \begin{pmatrix}z_1 \\
z_2 \end{pmatrix} \,.
\]
From (\ref{pauli}), using relations (\ref{fpb}) one easily obtains
\[
\{L_\mu,L_\nu\}= \frac i 4 y^*[\sigma_\mu, \sigma_\nu]y \ .
\]
Then relations (\ref{pb1})--(\ref{pb2}) can be immediately
verified using the well-known commutation relations between
matrices $\sigma$.
\end{proof}


Let us consider the subgroups of ${\cal G}_l$ generated by $L_i$,
with $i=1,2,3$. According to (\ref{sub}) we have
\begin{equation*}\begin{split}
\frac {d Y(\tau)}{d\tau} &=\{Y(\tau),L_i\}= -\frac i 2
\sigma_i Y(\tau) \,,\\
\frac {d Z_3(\tau)}{d\tau} &=\{Z_3(\tau),L_i\}= 0 \,,\end{split}
\end{equation*}
where
\[
Y(\tau)= \begin{pmatrix}Z_1(\tau) \\
Z_2(\tau) \end{pmatrix} \,.
\]
Hence
\begin{equation*} \begin{split}
Y(\tau) &= \exp\Big(-\frac i 2 \tau\sigma_i\Big)y \,,\\
Z_3(\tau) &= z_3 \ . \end{split}
\end{equation*}
Similarly, for the two subgroups generated by $L_4$ and $L_5$ we
have respectively
\begin{equation*} \begin{split}
Y(\tau) &= \exp\Big(-\frac i 2 \tau\Big)y \,,\\
Z_3(\tau) &= z_3 \,,\end{split}
\end{equation*}
and
\begin{equation*} \begin{split}
Y(\tau) &= y \,,\\
Z_3(\tau) &= \exp(i\tau)z_3 \ .
\end{split} \end{equation*}

From these formulas one can easily derive the general form of the
elements of ${\cal G}_{l}$.
\begin{prop}
For $l_1= l_2\neq l_3$ the group ${\cal G}_{l}$ is made by all the
transformations of the form
\begin{equation} \label{unit} \begin{split}
Y &= U y \,,\\
Z_3 &= e^{i\phi} z_3 \,, 
\end{split}
\end{equation}
where $U$ is a unitary $2\times 2$ matrix and $\phi \in
\mathbb{R}$. Hence the group ${\cal G}_{l}$ is isomorphic to $U(2)
\times U(1)$.
\end{prop}


Using (\ref{pauli}), we find that under a transformation of the
form (\ref{unit}) the functions (\ref{2gcom}) behave as
\begin{equation} \label{ellei}\begin{split}
L_i&\mapsto -\frac 12 Y^*\sigma_i Y = -\frac 12 y^*U^*\sigma_i U y
= \sum_{j=1}^3 R_{ij}(U)L_j
\qquad \text{for } i=1,2,3 \,, \\
L_4&\mapsto -\frac 12 Y^*\sigma_4 Y = -\frac 12 y^*U^*
\sigma_4 U y = L_4 \,,\\
L_5&\mapsto\bar Z_3 Z_3 = \bar z_3 z_3 =L_5 \,, \end{split}
\end{equation}
where $R_{ij}(U)$ are the elements of the $3 \times 3$ matrix
$R(U)\in SO(3)$ which is associated with $U$ by the standard
three-dimensional real representation $R: U(2)\to SO(3)$. The
kernel of this representation is made of all the matrices $U\in
U(2)$ of the form $U=e^{i\psi}E$, with $\psi\in \rn$.

Using (\ref{f1pab}) it is easy to see that any real polynomial in
involution with $F_1$ can be obtained as an algebraic combination
of the elements of a finite set. For $l_1=p>0$, $l_3=q>0$, this
set includes, besides the five second degree polynomials $L_\mu$,
the polynomials ${\rm Re}\{z_1^{q-s}z_2^s\bar z_3^{p}\}$ and ${\rm
Im}\{z_1^{q-s}z_2^s\bar z_3^p\}$ for $0\leq s \leq q$. If $l_1=
p<0$, one can write analogous formulas with $z_3^{-p}$ in place of
$\bar z_3^{p}$. Defining (for $p>0$)
\begin{equation} \label{als}
A_{q,s}=\frac {\bar z_1^{q-s}\bar z_2^s z_3^p} {\sqrt{(q-s)!s!}}
\,,
\end{equation}
we can write for any $a_s,b_s \in \mathbb{R}$
\[
a_s {\rm Re}\{z_1^{q-s}z_2^s\bar z_3^p\} +b_s{\rm
Im}\{z_1^{q-s}z_2^s\bar z_3^p\}= {\rm Re}\{c_s A_{q,s}\} \,,
\]
with
\[
c_s= \sqrt{(q-s)!s!}(a_s+ib_s) \in \cn\ .
\]
Note that
\[
\sum_{s=0}^q A_{q,s}^2 = \frac 1 {q!} \left(\bar z_1^2+ \bar z_2^2
\right)^q z_3^{2p}= \frac R{q!} \,,
\]
where $R$ is the function defined in (\ref{r3}).

The transformation properties of the $A_{q,s}$ under the action of
the group ${\cal G}_{l}$ are determined by their Poisson brackets
with the functions $L_\mu$. A direct calculation provides
\begin{equation} \label{ellemua}
\{iL_\mu,A_{q,s}\}=\sum_{h=0}^q \left(J_{q,\mu}\right)_{hs}
A_{q,h} \,,
\end{equation}
with
\begin{align}
\left(J_{q,1}\right)_{hs} &=\frac 12 \left[ \delta_{h-1,s}
\sqrt{(q-s)h} +
\delta_{s-1,h} \sqrt{(q-h)s} \right] \label{j1} \,,\\
\left(J_{q,2}\right)_{hs} &=\frac 1{2i} \left[ \delta_{h-1,s}
\sqrt{(q-s)h} -
\delta_{s-1,h}\sqrt{(q-h)s} \right] \,,\\
\left(J_{q,3}\right)_{hs} &=\left(\frac q 2 -h \right)\delta_{hs}
\label{j3} \,,\\
\left(J_{q,4}\right)_{hs} &=\frac q 2 \delta_{hs} \,,\\
\left(J_{q,5}\right)_{hs} &=p\delta_{hs} \,. \label{j5}
\end{align}
It follows that
\[
\{iL_\mu,\sum_{s=0}^q c_s A_{q,s}\}=\sum_{s=0}^q c^\prime_s
A_{q,s}\,,
\]
where
\[
c'_s=\sum_{r=0}^q\left(J_{q,\mu}\right)_{sr}c_r \ ,
\]
or in compact notation $c^\prime=J_{q,\mu}\cdot c$. The numerical
coefficients on the right-hand side of (\ref{als}) have been
chosen in such a way that the matrices $J_{q,\mu}$ are hermitian.
Other properties of these matrices, which can be directly verified
using formulas (\ref{j1})--(\ref{j5}), follow from simple general
considerations. Since $F_1 = qL_5 -2pL_4$, the relation
$\{F_1,A_{q,s}\}=0$ implies $q J_{q,5} =2pJ_{q,4}$. Furthermore,
as a consequence of the Jacobi identity we have for any function
$G$
\[
\{L_\mu,\{L_\nu,G\}\}- \{L_\nu,\{L_\mu,G\}\} = \{\{ L_\mu,L_\nu
\},G\} \ .
\]
By applying the previous equation with $G= \sum_{s=0}^q c_s
A_{q,s}$, one can see that the commutation relations between the
matrices $J_{q,\mu}$ are determined by the Poisson brackets
(\ref{pb1})--(\ref{pb2}) between the functions $L_\mu$. We have
accordingly
\begin{equation}
[J_{q,i}, J_{q,j}]=i\sum_{k=1}^3\varepsilon_{ijk} J_{q,k}
\label{cr1}
\end{equation}
for $i,j=1,2,3$, and
\begin{equation}
[J_{q,\mu}, J_{q,4}]=[J_{q,\mu}, J_{q,5}]=0 \label{cr2}
\end{equation}
for $\mu=1,\ldots,5$. Looking at formulas (\ref{j1})--(\ref{j3})
one can actually recognize the well-known $(q+1)$-dimensional
irreducible representation of $su(2)$, which satisfies the
relation
\[
\sum_{i=1}^3 J_{q,i}^2 =\frac{q(q+2)}4 \,E \ ,
\]
$E$ being the identity matrix of rank $q+1$. It follows that the
functions $A_{q,s}$ belong to the corresponding irreducible
representation of $U(2)$.


In the following of the present section, we shall always suppose
that $l_1=l_2= p>0$, $l_3=q >0$, $q\neq p$. However, using the
methods which we already applied in section \ref{closc}, the
results can be easily generalized to all cases in which
$|l_1|=|l_2|\neq |l_3|$.

\begin{thm} \label{teo1}
Any simple integrable set is canonically equivalent to a set $F=
(F_1,F_2,F_3)$ with
\begin{equation}
F_2= d_1 I_1+ d_2 I_2 \ , \label{effe2}
\end{equation}
where $d_1\in \mathbb{N}$ and $d_2\in \mathbb{Z}$ do not have
common divisors, and $|d_2|\leq d_1$.
\end{thm}
\begin{proof}
If $\deg F_2 =2$, according to proposition \ref{ll} we can write
$F_2= \sum_{\mu=1}^5 \alpha_\mu L_\mu$, with $\alpha_\mu \in \rn$.
By replacing $F_2$ with $F_2-(\alpha_5/q)F_1$, we obtain a
functionally equivalent set such that $F_2$ does not contain
$L_5=I_3$. Furthermore, by performing an appropriate
transformation of ${\cal G}_{l}$, according to the first of
(\ref{ellei}) we can operate a rotation on the $L_i$, $i=1,2,3$,
in such a way to eliminate from $F_2$ the terms proportional to
$L_1$ and $L_2$. After these operations we obtain a canonically
equivalent set with
\begin{equation} \label{f2}
F_2=\beta_3 L_3 + \beta_4 L_4 = \gamma_1I_1 +\gamma_2I_2 \ ,
\end{equation}
where $\beta_3=\sqrt{\alpha_1^2+\alpha_2^2+\alpha_3^2}$, $\beta_4=
\alpha_4+ 2p\alpha_5/q$, $\gamma_1 =-(\beta_3 +\beta_4)/2$,
$\gamma_2=(\beta_3-\beta_4)/2$. We further note that, applying
when necessary the canonical transformation that interchanges
$z_1$ and $z_2$, we can always ensure that $|\gamma_1|\geq
|\gamma_2|$.

In general $F_3$ can be written in the form
\[
F_3= \sum_{(a,b)\in K}{\rm Re} \{c_{a,b} P_{a,b}\} \ ,
\]
where $P_{a,b}:= z_1^{a_1}\bar z_1^{b_1} z_2^{a_2}\bar z_2^{b_2}
z_3^{a_3}\bar z_3^{b_3}$, $c_{a,b}\in \cn\ \forall\, (a,b)\in K$,
and $K \subset \mathbb{Z}_+^3 \times \mathbb{Z}_+^3$ is a finite
set such that
\begin{equation}
p(a_1 -b_1 +a_2-b_2) +q(a_3 -b_3)= 0 \quad \forall\, (a,b)\in K\,.
\label{aibi}
\end{equation}
From (\ref{f2}) we have in general
\[
\{F_2, P_{a,b}\} = -i[\gamma_1(a_1-b_1) +\gamma_2(a_2-b_2)]
P_{a,b} \ ,
\]
so that the condition $\{F_2,F_3\}=0$ implies
\begin{equation} \label{f2piabi}
\gamma_1(a_1-b_1)+\gamma_2(a_2-b_2)=0 \quad \forall\, (a,b)\in
K\,.
\end{equation}

If $a_1 -b_1= a_2-b_2 =0$, then from (\ref{aibi}) we obtain
$a_3-b_3=0$ and ${\rm Re} \{c_{a,b} P_{a,b}\}= ({\rm Re}\,c_{a,b})
I_1^{a_1}I_2^{a_2}I_3^{a_3}$. Hence, if $a_1 -b_1= a_2-b_2 =0\
\forall \, (a,b)\in K$, we have clearly $(F_1,F_2,F_3) \approx
(I_1,I_2,I_3) \approx (F_1,I_1,I_3)$, where $\approx$ denotes
functional equivalence. This means that the thesis of the theorem
holds with $d_1=1$, $d_2=0$ in formula (\ref{effe2}). On the other
hand, if there exists $(a,b)\in K$ such that, say, $a_2-b_2\neq
0$, it follows from (\ref{f2piabi}) that $\gamma_2/\gamma_1=
-(a_1-b_1)/(a_2-b_2) \in \mathbb{Q}$. Since one can always
redefine $F_2$ by multiplying it by a constant, the thesis follows
immediately from (\ref{f2}).
\end{proof}

It follows from Theorem \ref{teo1} that any simple integrable set
is canonically equivalent to a set of the form (\ref{fint1}), with
$r=(d_1, d_2, 0)$. In particular, any set of the form
(\ref{fint2}) or (\ref{fint3}) can be converted into the form
(\ref{fint1}) via a symplectic transformation. This can be checked
directly, by observing that under a transformation of the form
(\ref{unit}), with
\[ 
U= \exp\big(i\frac \pi 4 \sigma_1\big)= \frac {1+i\sigma_1} {\sqrt
2} \,, \qquad \phi= 0 \,,
\] 
we have according to (\ref{ellei}) $P_{12}\mapsto I_1-I_2$.
Similarly, for
\[ 
U= \exp\big(-i\frac \pi 4 \sigma_2\big)= \frac {1-i\sigma_2}
{\sqrt 2} \,, \qquad \phi= 0 \,,
\] 
we have $Q_{12}\mapsto I_1-I_2$.

When $r=(d_1, d_2, 0)$, the vector $m$ in (\ref{fint1}) must
satisfy the conditions
\begin{align*}
p(m_1+ m_2) +qm_3 &=0\,, \\
d_1 m_1+ d_2 m_2 &=0 \ .
\end{align*}
We see that, if $d_1=d_2$, we can take $m= (-1, 1, 0)$, which
means that $\im R_{m}=2P_{12}$. For $-d_1\leq d_2<d_1$ we can take
instead $m= (-d_2 h, d_1 h, -pk)$, where $h$ and $k$ are two
positive integers without common divisors such that
$q/(d_1-d_2)=h/k$.

We list in Table \ref{tab1} some examples of simple integrable
sets, which are obtained in this way for $l=(1,1,2)$.
\begin{table}[h]
\centering
\begin{tabular}{l|l|c}
$F_2$ & $F_3$ & degree\\
\hline
$I_1$ & Re$\{z_2^2\bar z_3\}$ & 3\\
$I_1-I_2$ & Re$\{z_1z_2\bar z_3\}$ & 3 \\
$3I_1 +I_2$ & Re$\{\bar z_1 z_2^3\bar z_3\}$ & 5 \\
$3I_1-I_2$ & Re$\{z_1z_2^3\bar z_3^2\}$ & 6 \\
$2I_1+I_2$ & Re$\{\bar z_1^2z_2^4\bar z_3\}$ & 7 \\
$5I_1+I_2$ & Re$\{\bar z_1z_2^5\bar z_3^2\}$ & 8 \\
$5I_1+3I_2$ & Re$\{\bar z_1^3z_2^5\bar z_3\}$ & 9 \\
$5I_1-I_2$ & Re$\{z_1z_2^5\bar z_3^3\}$ & 9
\end{tabular}
\caption{\label{tab1} Examples of simple integrable sets for $l=
(1,1,2)$.}
\end{table}

\subsection{A non-simple integrable algebra}\label{nons}

In the preceding section we have shown how one can construct
simple integrable sets of arbitrarily high degree. On the other
hand, we do not know of any general method to obtain non-simple
integrable sets, that is integrable sets $F$, with $2< \deg
F_2\leq \deg F_3$, such that there exists no functionally
equivalent set $F'$ with $\deg F_2'=2$. Nevertheless, we are able
to exhibit a concrete example of non-simple integrable set for
$l=(1,1,2)$. Let us define
\begin{align*}
D_0&={\rm Re}\{\bar z_1^2 z_3\}=\sqrt{2}{\rm Re}\{A_{2,0}\}\,,
&& C_0={\rm Im}\{\bar z_1^2 z_3\}=\sqrt{2}{\rm Im}\{A_{2,0}\}\,, \\
D_1&={\rm Re}\{\bar z_1\bar z_2 z_3\}={\rm Re}\{A_{2,1}\}\,, &&
C_1={\rm Im}\{\bar z_1\bar z_2 z_3\}={\rm Im}\{A_{2,1}\}\,, \\
D_2&={\rm Re}\{\bar z_2^2 z_3\}=\sqrt{2}{\rm Re}\{A_{2,2}\}\,, &&
C_2={\rm Im}\{\bar z_2^2 z_3\}=\sqrt{2}{\rm Im}\{A_{2,2}\}\,,
\end{align*}
and consider the functions
\begin{equation} \label{nonsimple}
F_1=I_1 +I_2 +2I_3 \,,\qquad F_2=C_0+2C_2 \,,\qquad
F_3=2C_0^2+I_1M_3^2\,,
\end{equation}
where $M_3:= P_{12}=  2{\rm Im}\{\bar z_1 z_2\}$. Clearly
$\{F_1,F_2\}=\{F_1,F_3\}=0$. Furthermore, by exploiting formulas
(\ref{ellemua})--(\ref{j5}), or by direct computation, one easily
finds
\begin{align}
&\{M_3,C_0\}=-2C_1\,,
&&\{M_3,C_2\}=2C_1 \,,\\
&\{I_1,C_0\}= 2D_0\,, &&\{I_1,C_2\}= 0 \,.\label{extlie}
\end{align}
We also have
\[
\{C_0,C_2\}=\frac i 4(\bar z_1^2 z_2^2-z_1^2 \bar z_2^2) = -\frac
{M_3 N_3}4 \ ,
\]
where $N_3:= Q_{12}=2{\rm Re}\{\bar z_1 z_2\}$. Using these
relations one finds
\[
\{F_2,F_3\}= 2M_3(C_0 N_3 -D_0 M_3- 2I_1C_1)= 0 \ ,
\]
where the last equality follows from
\begin{equation*} \begin{split}
C_0 N_3-D_0 M_3&=2{\rm Im}\{\bar z_1^2 z_3 \} {\rm Re}\{z_1 \bar
z_2\} + 2{\rm Re}\{\bar z_1^2 z_3\} {\rm Im}\{
z_1 \bar z_2\} \\
&=2{\rm Im}\{z_1 \bar z_1^2 \bar z_2  z_3\}=2I_1C_1 \ .
\end{split} \end{equation*}

It is easy to see that the set $F=(F_1, F_2, F_3)$ is functionally
independent. Therefore $F$ is an integrable set with $\deg F_2=3$,
$\deg F_3 =6$.  This implies that any system with hamiltonian
$H=f(F)$, where $f$ is an arbitrary function of three variables,
is also integrable. It is also easy to obtain a similar integrable
set for any other $l$ such that $|l_1|= |l_2|=1$, $|l_3|=2$. By
means of symplectic transformations of the form (\ref{unit}), we
can then obtain a whole class of non-simple integrable sets which
are canonically equivalent to $F$. However, we do not know at
present of any other class of non-simple integrable sets, either
for $l=(\pm 1, \pm 1, \pm 2)$ or for any other value of $l$.
Attempts to discover other examples, also with the aid of
specially made computer programs, have proven unsuccessful.

\subsection{Quantization} Let us now consider the problem of the
quantization of the integrable sets that we have obtained in
sections \ref{case} and \ref{nons}. According to proposition 
3.1 (case 1) of \cite{part2}, all simple integrable sets can be
straightforwardly quantized by symmetrization with respect to the
operators $(\hat z, \hat z^*)$ which were defined in section
\ref{qosc}. Therefore any quantum system with hamiltonian operator
$\hat H=f(\hat F)$, where $\hat F= F^{\rm sym}$ is the
symmetrization of a simple integrable set and $f$ is an arbitrary
function of three variables, is quasi-integrable.

As regards the non-simple algebra (\ref{nonsimple}), let $F_i^{\rm
sym}$, $i=1,2,3$, be the operators which are obtained by
symmetrization of the functions $F_i$ with respect to the
operators $(\hat z, \hat z^*)$. Using proposition 
4.1 (case a) of \cite{part2} we obtain that
\[
[F_1^{\rm sym},F_2^{\rm sym}] =[F_1^{\rm sym}, F_3^{\rm sym}]=0
\,.
\]
However this proposition 
does not allow us to evaluate $[F_2^{\rm sym}, F_3^{\rm sym}]$,
since both the involved operators have degree higher than 2. Let
us then make use of the general formula of Moyal brackets.
With the conventions adopted in this section,
this formula for two generic polynomials $H(z, \bar z)$ and $F(z,
\bar z)$ can be written as $[H^{\rm sym}, F^{\rm sym}]=G^{\rm
sym}$, where
\[
G= \sum_{k\in \nn}\ \sum_{|\alpha+ \beta|=2k+1}\frac
{(-1)^{|\beta|}} {2^{2k}\alpha ! \beta !}\,\frac
{\partial^{|\alpha + \beta|} H}{\partial z^\alpha
\partial \bar z^\beta}\, \frac{\partial^{|\alpha + \beta|}F} {\partial
z^\beta \partial \bar z^\alpha} \ .
\]
By applying this formula we obtain
\[
[F_2^{\rm sym},F_3^{\rm sym}] =\frac 52 i\hat D_0\,,
\]
where $\hat D_0$ is the standard quantization of the function
$D_0$ (symmetrization is in this case unnecessary). On the other
hand, we have from (\ref{extlie}) and from proposition 
3.1 of \cite{part2} that
\[
[F_2^{\rm sym}, \hat I_1] =-i\{F_2,I_1\}^{\rm sym} =2i\hat D_0\, .
\]
Since $[F_1^{\rm sym},\hat I_1] =0$, it follows from the two above
equalities that the three operators
\[
\hat F_1 =F_1^{\rm sym}\,, \qquad \hat F_2 =F_2^{\rm sym}\,,
\qquad \hat F_3 =F_3^{\rm sym}-\frac 54 \hat I_1
\]
are pairwise in involution. Since these operators are also
quasi-independent, one concludes that the set $\hat F= (\hat F_1,
\hat F_2, \hat F_3)$ is a quasi-integrable set of operators.
Hence, any quantum system with hamiltonian operator $\hat H=f(\hat
F)$, where $f$ is an arbitrary function of three variables, is
quasi-integrable. We have thus shown that it is possible to
quantize our example (\ref{nonsimple}) of non-simple integrable
set. However, to this purpose the general procedure of
quantization by symmetrization needs to be modified. The
modification is represented by the introduction of the term
$-(5/4) \hat I_1$ in the expression of the operator $\hat F_3$.
This fact presents an analogy with the situation for a free
quantum rigid body in 6-dimensional space \cite{part3}, and also
with some results that have already been obtained for other types
of integrable quantum systems \cite{hieta, hietagr}.

\end{document}